\documentclass[11pt,a4paper]{article}
\usepackage[a4paper,top=2.35cm,bottom=2.35cm,left=1.9cm,right=1.9cm,
             headheight=15pt,headsep=0.55cm,footskip=0.9cm]{geometry}
\usepackage{amsmath,amssymb,amsthm,mathtools,bm}
\usepackage{array}
\usepackage{aliascnt}
\usepackage{authblk}
\usepackage{balance}
\usepackage{enumitem}
\usepackage{microtype}
\usepackage{fancyhdr}
\usepackage[numbers,sort&compress]{natbib}
\usepackage{hyperref}
\usepackage[nameinlink,noabbrev,capitalize]{cleveref}

\hypersetup{
  hidelinks,
  pdftitle={A comment on local hypercube inequalities: analytic proofs in odd and even dimensions},
  pdfsubject={Analytic proofs of the odd- and even-dimensional local hypercube inequalities},
  pdfkeywords={higher-dimensional partitions, charge functions, local inequality, Boolean lattice, coordinate flip, sign-reversing injection}
}

\setlist[itemize]{leftmargin=2.0em,itemsep=0.2em,topsep=0.35em}
\setlist[enumerate]{leftmargin=2.2em,itemsep=0.22em,topsep=0.35em}
\allowdisplaybreaks[2]

\theoremstyle{plain}
\newtheorem{theorem}{Theorem}[section]
\newaliascnt{lemma}{theorem}
\newtheorem{lemma}[lemma]{Lemma}
\aliascntresetthe{lemma}
\newaliascnt{proposition}{theorem}
\newtheorem{proposition}[proposition]{Proposition}
\aliascntresetthe{proposition}
\newaliascnt{corollary}{theorem}
\newtheorem{corollary}[corollary]{Corollary}
\aliascntresetthe{corollary}
\theoremstyle{definition}
\newaliascnt{definition}{theorem}
\newtheorem{definition}[definition]{Definition}
\aliascntresetthe{definition}
\newaliascnt{example}{theorem}
\newtheorem{example}[example]{Example}
\aliascntresetthe{example}
\theoremstyle{remark}
\newaliascnt{remark}{theorem}
\newtheorem{remark}[remark]{Remark}
\aliascntresetthe{remark}

\newcommand{\calI}{\mathcal I}
\newcommand{\calP}{\mathcal P}
\newcommand{\calN}{\mathcal N}
\newcommand{\one}[1]{\mathbf 1_{\{#1\}}}
\newcommand{\odd}{\mathrm{odd}}
\newcommand{\even}{\mathrm{even}}
\newcommand{\Add}{\operatorname{Add}}
\newcommand{\Rem}{\operatorname{Rem}}

\title{\vspace{-1.15cm}\bfseries
A Comment on Local Hypercube Inequalities\\[0.35em]
\large Analytic Proof 
\large for Higher-Dimensional Partition Charge Functions}
\author[a]{Keyou Zhuo}
\author[a]{Tian-Shun Chen\textsuperscript{*}}
\author[a,b,c]{Kilar Zhang\textsuperscript{*}}
\affil[a]{\small Department of Physics and Institute for Quantum Science and Technology,
Shanghai University, Shanghai 200444, China}
\affil[b]{\small Shanghai Key Lab for Astrophysics, Shanghai 200234, China}
\affil[c]{\small Shanghai Key Laboratory of High Temperature Superconductors,
Shanghai University, Shanghai 200444, China}
\affil[]{\small\texttt{zhuokeyou@shu.edu.cn};
\texttt{cts2003912@shu.edu.cn}; \texttt{kilar@shu.edu.cn}\ 
(\textsuperscript{*}corresponding authors)}
\date{\vspace{-1.0em}}

\begin{document}
\maketitle
\thispagestyle{plain}
\vspace{-2.4em}

\begin{abstract}
{This comment gives dimension-independent analytic proofs of the local inequalities arising in the odd- and even-dimensional constructions of higher-dimensional partition charge functions.  The prior works established these inequalities by exhaustive computation in low dimensions and tested them numerically in selected higher dimensions.
}
\end{abstract}

\begin{center}
\end{center}

\twocolumn

\section{Introduction}

{This comment addresses an analytic gap in two local hypercube arguments. \cite{XiangEtAl2026Odd} reduced the pole-matching problem for their odd-dimensional charge
functions to a local hypercube assertion.  Its equality branch was proved analytically,
whereas the inequality branch was checked by exhaustive enumeration in
dimension five and tested by Monte Carlo sampling in dimensions seven and nine.  Then~\cite{FengChenZhang2025All} formulated the even-dimensional analogue.  \cite{FengChenZhang2025All} proved
the equality branch analytically, enumerated the inequality branch in dimension six, and supplied
Monte Carlo evidence in dimension eight.  The aim of the present note is to replace these
dimension-specific computations by analytic proofs of the odd- and even-dimensional
inequality branches.}

{The main point is local.  Once the hypercube configuration is separated from the global
partition background, a partition contained in a local \(0/1\) hypercube is exactly a
down-set of a Boolean lattice.  The local pole orders appearing in both papers translate into
signed statistics of the ranks of this down-set, with certain positive terms selected by a
one-step saturation condition.  A fixed coordinate flip then gives an explicit sign-reversing
injection.  The proof uses the Boolean lattice as a common language for the two inequalities;
the connections with monotone Boolean functions, simplicial complexes, and poset matchings
are secondary to this local calculation.}

{We first derive the Boolean layer statistics \(E_r,C_r\) from the local factors of
the two prior works \cite{XiangEtAl2026Odd, FengChenZhang2025All}, then apply one fixed coordinate flip to every intermediate rank, and
finally treat the odd and even full-dimensional endpoints separately.  Combined with
Lemmas~1--4 of \cite{XiangEtAl2026Odd}, the odd-dimensional result supplies the missing local analytic
input in their argument.  In the even-dimensional case, it supplies the local inequality
required by the strategy adopted in Section~V of \cite{FengChenZhang2025All}.}

\section{Boolean-lattice model}

\subsection{Order ideals and equivalent models}

{Let \(D\) be a nonempty finite set and put
\[
 d:=|D|,\qquad B_D:=2^D.
\]
Ordered by inclusion, \(B_D\) is the Boolean lattice of rank \(d\), with minimum
\(\varnothing\) and maximum \(D\).}

{\begin{definition}[Down-set]
A family \(\calI\subseteq2^D\) is a \emph{down-set}, or order ideal, if
\begin{equation}
 A\in\calI,\quad B\subseteq A
 \quad\Longrightarrow\quad B\in\calI.
 \label{eq:downset-en}
\end{equation}
The void ideal \(\calI=\varnothing\) is allowed.  Every nonempty ideal contains
\(\varnothing\).
\end{definition}
}

{Identifying \(A\) with its indicator vector \(\mathbf1_A\in\{0,1\}^D\), a down-set is exactly
the truth set of a coordinatewise decreasing Boolean function
\[
 f:\{0,1\}^D\longrightarrow\{0,1\},
 \qquad f(\mathbf1_A)=1\Longleftrightarrow A\in\calI.
\]}

A nonempty down-set is also precisely the face set of an abstract simplicial complex on a
subset of \(D\).  In the oriented hypercube graph, \(A\) is a vertex and
\(A\leftrightarrow A\triangle\{i\}\) is an \(i\)-edge.  Boolean-lattice ideals, decreasing
Boolean functions, simplicial complexes, and oriented hypercubes therefore give four
representations of the same local object.

{The number of these objects is the Dedekind number \(M(d)\) \cite{Dedekind1897, Chen:2026cfo}.  For example,
\[
 M(1),\ldots,M(6)=3,6,20,168,7581,7828354.
\]}
Their rapid growth explains why exhaustive verification is feasible in small dimensions but
cannot replace a structural proof.

\subsection{Reverse rank, saturation, and layer counts}

{We count ranks downward from the maximum:
\begin{equation}
 \rho(A):=|D\setminus A|=d-|A|.
 \label{eq:rho-en}
\end{equation}
Thus \(\rho(D)=0\) and \(\rho(\varnothing)=d\).  An element \(D\setminus\{j\}\) has reverse
rank one and is a coatom of \(B_D\).  A one-element set \(\{j\}\) is called a singleton.}

{\begin{definition}[Saturated element]
An element \(A\in\calI\) is \emph{saturated} if every immediate upper neighbor toward \(D\)
also lies in \(\calI\):
\begin{equation}
 A\cup\{i\}\in\calI
 \qquad\text{for every }i\in D\setminus A.
 \label{eq:saturated-en}
\end{equation}
\end{definition}
}

{Define
\begin{align}
 E_r(\calI)
 &:=\#\{A\in\calI:\rho(A)=r\},
 \label{eq:Er-en}\\
 C_r(\calI)
 &:=\#\{A\in\calI:\rho(A)=r,\ A\text{ is saturated}\}.
 \label{eq:Cr-en}
\end{align}}
Clearly \(0\le C_r\le E_r\).  Saturation is only a one-step condition: it does not assert
that the entire interval above \(A\) lies in \(\calI\).  In simplicial language it says that
the link of \(A\) contains every vertex in \(D\setminus A\), not that the link is a simplex.

\subsection{The standard signed statistic}

{Set
\begin{equation}
 \Phi(\calI)
 :=E_1
 -\sum_{\substack{2\le r\le d\\r\ {\rm even}}}E_r
 +\sum_{\substack{3\le r\le d\\r\ {\rm odd}}}C_r.
 \label{eq:Phi-en}
\end{equation}
Its positive and negative objects are
\begin{align}
 \calP_\Phi
 &:={}
 \{A\in\calI:\rho(A)=1\}
 \notag\\[-0.2em]
 &\quad\cup
 \left\{A\in\calI:
 \substack{A\text{ saturated},\\ \rho(A)\ge3\text{ odd}}
 \right\},
 \label{eq:Pphi-en}\\
 \calN_\Phi
 &:={}
 \{A\in\calI:\rho(A)\ge2\text{ even}\}.
 \label{eq:Nphi-en}
\end{align}
Consequently,
\begin{equation}
 \Phi(\calI)=|\calP_\Phi|-|\calN_\Phi|.
 \label{eq:Phi-card-en}
\end{equation}}

\subsection{Broader combinatorial context}

{A sign-reversing injection means an injective pairing that sends every object
with coefficient \(+1\) to an object with coefficient \(-1\).}  The construction is in the
spirit of the involution principle of~\cite{GarsiaMilne1981}.  On the Hasse graph of the
Boolean lattice, it pairs edges parallel to one fixed coordinate.  The proof below uses only
this explicit pairing; unmatched objects can occur only at the endpoint ranks.  Related
interpretations in discrete Morse theory may be found in~\cite{Forman1998,Kozlov2008},
although none of that machinery is needed here.

The statistic \(\Phi\) mixes ordinary rank numbers \(E_r\) with the saturation counts
\(C_r\), so the classical shadow inequalities~\cite{Kruskal1963,Katona1968} do not apply to
it directly.  Its value on the two
extremal ideals instead reduces to the usual alternating binomial cancellation
\cite{Stanley2012}.  These observations provide context for the construction; the proof
itself uses only the coordinate flip below.

\section{Coordinate-flip matching}

{This section proves the matching lemma used in both parity cases.  The matching acts on
the down-set associated with a single local hypercube.}

\begin{lemma}[A missing coatom]
\label{lem:missing-coatom-en}
If a down-set \(\calI\subseteq2^D\) is neither \(2^D\) nor \(2^D\setminus\{D\}\), then there is
\(j\in D\) such that
\begin{equation}
 D\setminus\{j\}\notin\calI.
 \label{eq:missing-coatom-en}
\end{equation}
\end{lemma}

\begin{proof}
If every coatom belonged to \(\calI\), then for any proper subset \(B\subsetneq D\) one could
choose \(j\in D\setminus B\) and use
\(B\subseteq D\setminus\{j\}\).  Downward closure would place every proper subset in
\(\calI\).  The only remaining choice would be whether \(D\) itself is present, giving
\(\calI=2^D\setminus\{D\}\) or \(\calI=2^D\).
\end{proof}

\begin{theorem}[Coordinate-flip injection]
\label{thm:flip-en}
Let \(\calI\subseteq2^D\) be a down-set, and suppose that
\(D\setminus\{j\}\notin\calI\) for some \(j\in D\).  Define
{\begin{equation}
 \tau_j(A):=A\triangle\{j\}
 =\begin{cases}
 A\setminus\{j\},&j\in A,\\
 A\cup\{j\},&j\notin A.
 \end{cases}
 \label{eq:toggle-en}
\end{equation}
}
Then \(\tau_j\) restricts to an injection
\begin{equation}
 \tau_j:\calP_\Phi\hookrightarrow\calN_\Phi.
 \label{eq:injection-en}
\end{equation}
In particular, \(\Phi(\calI)\le0\).
\end{theorem}

\begin{proof}
We separate four elementary steps.

\emph{Step 1: choose the missing direction.}
The hypothesis excludes the unique reverse-rank-one element not containing \(j\).

\emph{Step 2: define the flip.}
The map in \cref{eq:toggle-en} changes only coordinate \(j\), and
\begin{equation}
 \tau_j^2=\operatorname{id}_{2^D}.
 \label{eq:involution-en}
\end{equation}

\emph{Step 3: check the image of every positive object.}
Take \(A\in\calP_\Phi\).
\begin{enumerate}
  \item If \(j\in A\), then \(\tau_j(A)=A\setminus\{j\}\in\calI\) by downward closure.
        The reverse rank changes from an odd \(r\) to the even number \(r+1\ge2\), so the
        image belongs to \(\calN_\Phi\).  This includes \(r=1\).
  \item If \(j\notin A\) and \(\rho(A)=1\), then
        \(A=D\setminus\{j\}\), contrary to the missing-coatom hypothesis.  This case cannot
        occur.
  \item If \(j\notin A\) and \(A\) is saturated of odd reverse rank \(r\ge3\), then
        \(\tau_j(A)=A\cup\{j\}\in\calI\) by saturation.  Its reverse rank is the even number
        \(r-1\ge2\), so it belongs to \(\calN_\Phi\).
\end{enumerate}

\emph{Step 4: use involutivity.}
If \(\tau_j(A_1)=\tau_j(A_2)\), applying \(\tau_j\) once more and using
\cref{eq:involution-en} gives \(A_1=A_2\).  Thus \cref{eq:injection-en} is injective, and
\cref{eq:Phi-card-en} gives \(\Phi(\calI)\le0\).
\end{proof}

\begin{corollary}[Extremal dichotomy]
\label{cor:standard-dichotomy-en}
For every down-set \(\calI\subseteq2^D\),
\begin{equation}
 \Phi(\calI)=
 \begin{cases}
  1,&\calI\in\{2^D,2^D\setminus\{D\}\},\\
  \le0,&\text{otherwise}.
 \end{cases}
 \label{eq:standard-dichotomy-en}
\end{equation}
\end{corollary}

\begin{proof}
The nonextremal case follows from
\cref{lem:missing-coatom-en,thm:flip-en}.  In either extremal ideal,
\(E_r=\binom dr\), and every element actually counted by \(C_r\) is saturated.  Hence
\[
 \Phi(\calI)=\sum_{r=1}^d(-1)^{r+1}\binom dr=1.
\]
\end{proof}

\begin{remark}[Hypercube matching]
The Hasse graph of \(B_D\) is the \(d\)-dimensional hypercube.  The map \(\tau_j\) matches all parallel \(j\)-edges; \cref{thm:flip-en} retains only the edges starting at positive objects.
A downward edge is legal by the ideal property, an upward edge is legal by saturation,
and the missing coatom removes the only top obstruction.  Thus the proof gives the matching
explicitly and does not require an application of Hall's theorem.
\end{remark}

\begin{example}[A four-dimensional flip]
Let \(D=\{1,2,3,4\}\) and
\[
 \calI=2^{\{1,2,4\}}\cup2^{\{1,3\}}.
\]
Then \(D\setminus\{4\}=\{1,2,3\}\notin\calI\), so we use \(\tau_4\).  The positive
objects are \(\{1,2,4\}\), of reverse rank one, and the saturated singleton \(\{1\}\),
of reverse rank three.  They are sent to the distinct reverse-rank-two objects
\[
 \{1,2,4\}\longmapsto\{1,2\},
 \qquad
 \{1\}\longmapsto\{1,4\}.
\]
Here \(E_1=1\), \(E_2=4\), \(C_3=1\), and \(E_4=1\), hence
\(\Phi(\calI)=1-4+1-1=-3\).  This example shows both kinds of legal move: the first is
downward and uses the ideal property, while the second is upward and uses saturation.
\end{example}

\section{Boolean form of local formulas}

This section derives the Boolean layer counts \(E_r,C_r\) from the local factors and cluster
conventions of the two cited works, placing both inequalities in the setting of the preceding
coordinate-flip theorem.

\subsection{Partitions and generic projections}

{Write \([n]=\{1,\ldots,n\}\), and let \(\vec e_i\) be the \(i\)-th standard basis vector.
An \(n\)-dimensional partition is a subset
\(\Delta^{(n)}\subseteq\mathbb Z_{\ge0}^n\) satisfying the melting rule: if
\(\vec{\square}=\sum_{i=1}^n l_i\vec e_i\in\Delta^{(n)}\) and \(l_i>0\), then
\(\vec{\square}-\vec e_i\in\Delta^{(n)}\).}

{Let \(h_1,\ldots,h_n\) satisfy the Calabi--Yau relation
\begin{equation}
 \sum_{i=1}^n h_i=0.
 \label{eq:CY-en}
\end{equation}
For \(\vec{\square}=\sum_{i=1}^n l_i\vec e_i\in\mathbb Z_{\ge0}^n\) and
\(R\subseteq[n]\), put
\[
 c(\vec{\square}):=\sum_{i=1}^n l_i h_i,
 \qquad
 h_R:=\sum_{i\in R}h_i.
\]}
{We require genericity in the following integral sense:
\begin{equation}
 \begin{aligned}
  &\sum_{i=1}^n a_i h_i=0,\qquad a_i\in\mathbb Z
  \\[-0.15em]
  &\Longrightarrow
  (a_1,\ldots,a_n)=q(1,\ldots,1),
  \quad q\in\mathbb Z.
 \end{aligned}
 \label{eq:generic-en}
\end{equation}}
{An \emph{integral resonance} is an integer relation
\(\sum_i a_i h_i=0\) that is not an integer multiple of the Calabi--Yau relation
\(\sum_i h_i=0\).}  If such a resonance occurs, distinct boxes or clusters can project to
the same point, and the unit-multiplicity statistics below must be regrouped.

All results below are stated under \cref{eq:generic-en}.  The resonant case is not covered:
there the local pole order must retain the multiplicities of all direction sets with the same
projection.

{Define the addable and removable boxes of a partition \(\Delta\) by
\begin{align}
 \Add(\Delta)
 &:={}
 \left\{\vec{\square}\notin\Delta:
 \substack{\vec{\square}-\vec e_i\in\Delta\\
 \text{for every }i\text{ with }l_i>0}
 \right\},
 \label{eq:Add-en}\\
 \Rem(\Delta)
 &:={}
 \left\{\vec{\square}\in\Delta:
 \substack{\vec{\square}+\vec e_i\notin\Delta\\
 \text{for every }i\in[n]}
 \right\}.
 \label{eq:Rem-en}
\end{align}}
{Following the notation of the cited charge-function papers, for a target box
position \(\vec{\square}\) set
\begin{equation}
 \begin{aligned}
 G(\vec{\square}):=\{\Delta:\;&\exists \widetilde{\vec{\square}}\in
 \Add(\Delta)\cup\Rem(\Delta),\\[-0.2em]
 &c(\widetilde{\vec{\square}})=c(\vec{\square})\}.
 \end{aligned}
 \label{eq:G-en}
\end{equation}}

\subsection{Admissible hypercube ideals}

{Choose a set of coordinate directions
\(D=\{n_1,\ldots,n_d\}\subseteq[n]\).  In the notation of the cited papers, define
\begin{equation}
 \begin{aligned}
  \vec{\square}_A&:=\sum_{i\in A}\vec e_i,\qquad A\subseteq D,\\
  HC^{(d)}&:=HC^{(d)}\!\left(\vec 0,\{\vec e_{n_i}\}_{i=1}^{d}\right)
  =\{\vec{\square}_A:A\subseteq D\},\\
  \vec q_d&:=\vec{\square}_D=\sum_{i\in D}\vec e_i.
 \end{aligned}
 \label{eq:HC-en}
\end{equation}}
{The point \(\vec q_d\) is the vertex of the hypercube opposite its origin.
For a partition \(\Delta^{(n)}\subseteq HC^{(d)}\), let
\begin{equation}
 \calI:=\{A\subseteq D:\vec{\square}_A\in\Delta^{(n)}\}.
 \label{eq:encoding-en}
\end{equation}
The melting rule is exactly the assertion that \(\calI\) is a down-set.}

\begin{lemma}[Uniqueness of a target projection]
\label{lem:projection-en}
Let \(A\subseteq D\), \(R\subseteq[n]\), and \(1\le|R|<n\).  Under
\cref{eq:generic-en},
\begin{equation}
 c(\vec{\square}_A)+h_R=c(\vec q_d)
 \quad\Longleftrightarrow\quad
 R=D\setminus A.
 \label{eq:projection-en}
\end{equation}
\end{lemma}

\begin{proof}
Moving all terms to one side gives
\[
 \sum_{i=1}^n
 \bigl(\mathbf1_{\{i\in A\}}+\mathbf1_{\{i\in R\}}-
       \mathbf1_{\{i\in D\}}\bigr)h_i=0.
\]
Genericity forces every coefficient in parentheses to be one common integer \(q\).
If \(q=1\), then \(A=D\) and \(R=[n]\), contradicting \(|R|<n\).  If \(q=-1\), then
\(D=[n]\) and \(A=R=\varnothing\), contradicting \(|R|\ge1\).  Hence \(q=0\), and
coordinatewise comparison gives \(R=D\setminus A\).  The converse is immediate.
\end{proof}

\begin{proposition}[Classification of admissible ideals]
\label{prop:admissible-en}
Identify \(\calI\) with its partition inside \(HC^{(d)}\).  Then
\begin{align}
 d<n:\quad
 &\Delta^{(n)}\in G(\vec q_d)
 \notag\\[-0.2em]
 &\qquad\Longleftrightarrow
 \calI\in\{2^D\setminus\{D\},2^D\},
 \label{eq:admissible-lower-en}\\
 d=n:\quad
 &\Delta^{(n)}\in G(\vec q_d)
 \notag\\[-0.2em]
 &\qquad\Longleftrightarrow
 \calI\in
 \{2^D,2^D\setminus\{D\},\varnothing,\{\varnothing\}\}.
 \label{eq:admissible-full-en}
\end{align}
\end{proposition}

\begin{proof}
By \cref{eq:generic-en}, every nonnegative lattice point \(\vec{\square}\) with
\(c(\vec{\square})=c(\vec q_d)\) has the form
\(\vec{\square}=\vec q_d+q\vec E\),
{where \(\vec E:=\sum_i\vec e_i\).}  If \(d<n\), a negative \(q\) gives
a negative coordinate.  For \(q\ge1\), the point and at least one required predecessor lie
outside the local \(0/1\) hypercube; the point is neither removable from the present partition
nor addable to it.  Thus only \(q=0\), namely \(\vec q_d\), remains.  The target is addable exactly
when all coatoms are occupied and \(D\) is absent, which by downward closure is
\(\calI=2^D\setminus\{D\}\).  It is removable exactly when \(D\in\calI\), which forces
\(\calI=2^D\).

When \(d=n\), one also has \(\vec q_n=\vec E\), and \(q=-1\) yields the origin.  The origin is
addable exactly for the void partition \(\calI=\varnothing\), and removable exactly for the
one-origin partition \(\calI=\{\varnothing\}\).  All other \(q\) are excluded as above.
\end{proof}

\subsection{Star clusters and local potentials}

{Following the cluster notation of the cited works, let \(p=|S|+1\).  For a center
\(A\subseteq D\) and a set of distinct arm directions \(S\subseteq D\setminus A\), define
the \(p\)-box star cluster by
\[
 \phi_p(\vec{\square}_A,S):=
 \{\vec{\square}_A\}\cup
 \{\vec{\square}_{A\cup\{i\}}:i\in S\}.
\]
It contains \(|S|+1\) boxes and carries the projection parameter
\(c(\vec{\square}_A)+h_S\).  For \(|S|<n\), \cref{lem:projection-en} shows that a star
projects to the target \(\vec q_d\) only
when \(S=D\setminus A\).  The full-arm case \(|S|=n\) is included explicitly in
\cref{prop:multiplicity-one-en}.}

{Here \(u\) denotes the local spectral parameter.  Following the potential-function
notation of the cited papers, for \(\epsilon\in\{\odd,\even\}\) let
\(\omega^\epsilon_{0,\Delta^{(n)}}(\vec q_d)\) denote the pole order of the corresponding
local charge function at \(u=c(\vec q_d)\): a pole has positive order and a zero has negative
order.  The superscript \(\epsilon\) is used only to distinguish the odd- and even-dimensional
charge functions treated together in this comment.  We also
write \(\one{P}\) for the indicator of a statement \(P\).  The vacuum factor \(1/u\)
therefore contributes \(\one{d=n}\), since genericity and the Calabi--Yau relation give
\(c(\vec q_d)=0\) exactly in the full-dimensional case.}

\begin{proposition}[Multiplicity-one reduction]
\label{prop:multiplicity-one-en}
Assume \cref{eq:generic-en}, and let \(\Delta^{(n)}\subseteq HC^{(d)}\).  At the target
\(\vec q_d\),
the local factors are classified as follows.
\begin{enumerate}
  \item A factor \(u-h_R\) centered at \(\vec{\square}_A\) contributes if and only if
        \(R=D\setminus A\).  In particular, a denominator \(u-h_i\) contributes exactly
        at reverse rank one.
  \item In even dimension, a factor \(u+h_i\) centered at \(\vec{\square}_A\) contributes if and only if
        \(D\setminus A=[n]\setminus\{i\}\); these are precisely the reverse-rank-\(n-1\)
        contributions.
  \item A star factor centered at \(\vec{\square}_A\) contributes if and only if its arm set is
        \(D\setminus A\).  Such a star is contained in the partition if and only if
        \(A\) is saturated.  This includes the full-arm endpoint \(D=[n]\),
        \(A=\varnothing\).
\end{enumerate}
Every contribution described above has multiplicity one.
\end{proposition}

\begin{proof}
The first assertion is \cref{lem:projection-en}.  For the second, the Calabi--Yau relation
gives \(-h_i=h_{[n]\setminus\{i\}}\), so the same lemma applies with
\(R=[n]\setminus\{i\}\).  A cluster contained in \(HC^{(d)}\) can use only directions in
\(D\setminus A\).  If its arm set has size less than \(n\), applying
\cref{lem:projection-en} to the projection parameter forces that set to be all of
\(D\setminus A\).  If it has size \(n\), then necessarily \(D=[n]\),
\(A=\varnothing\), and the same conclusion holds because \(h_{[n]}=0=c(\vec q_n)\).
Containment of the star is exactly the saturation condition \cref{eq:saturated-en}.
Finally, \cref{eq:generic-en} rules out any second direction set with the same target
projection, which proves the multiplicity statement.
\end{proof}

\subsubsection*{Odd-dimensional factors}
{Let \(n=2K+1\).  In the local-factor convention of
\cite{XiangEtAl2026Odd},
\begin{equation}
 \begin{aligned}
 \varphi^{\odd}_1(u)
 &=\frac{\displaystyle
   \prod_{m=1}^{K}\prod_{\substack{R\subseteq[n]\\|R|=2m}}(u-h_R)}
  {\displaystyle\prod_{i=1}^n(u-h_i)},\\[0.4em]
 \varphi^{\odd}_{2m}(u)
 &=u^{-1},\qquad 2\le m\le K.
 \end{aligned}
 \label{eq:odd-factors-en}
\end{equation}
}
By \cref{prop:multiplicity-one-en}, the single-box denominator gives \(+E_1\), the even
direction sets in the numerator give zeros on the even reverse ranks, and an even-box star
has an odd number of arms and gives a pole at a saturated odd reverse rank.  Therefore
\begin{equation}
 \begin{aligned}
 \omega^{\odd}_{0,\Delta^{(n)}}(\vec q_d)
 ={}&\one{d=n}+E_1\\
 &-\sum_{\substack{2\le r\le\min(d,n-1)\\r\ {\rm even}}}E_r\\
 &+\sum_{\substack{3\le r\le\min(d,n-2)\\r\ {\rm odd}}}C_r.
 \end{aligned}
 \label{eq:odd-potential-en}
\end{equation}

\subsubsection*{Even-dimensional factors}
{Let \(n=2K\ge4\).  In the convention of
\cite{FengChenZhang2025All},
\begin{equation}
 \begin{aligned}
 \varphi^{\even}_1(u)
 &=\frac{\prod_{i=1}^n(u+h_i)}{\prod_{i=1}^n(u-h_i)}\\[-0.1em]
 &\quad\times
   \prod_{m=1}^{K-1}
   \prod_{\substack{R\subseteq[n]\\|R|=2m}}(u-h_R),\\[0.3em]
 \varphi^{\even}_{2m}(u)&=u^{-1},
 \qquad 2\le m\le K-1,\\
 \varphi^{\even}_{n}(u)&=u^{-2},
 \qquad
 \varphi^{\even}_{n+1}(u)=u^2.
 \end{aligned}
 \label{eq:even-single-en}
\end{equation}
}
By \cref{prop:multiplicity-one-en}, the additional factors \(u+h_i\) give zeros at reverse
rank \(n-1\).  The \(n\)-box and \((n+1)\)-box stars give, respectively, a double pole and a
double zero at the two endpoint ranks.  Hence
\begin{align}
 \omega^{\even}_{0,\Delta^{(n)}}(\vec q_d)
 ={}&\one{d=n}+E_1
 \notag\\
 &-\sum_{\substack{2\le r\le\min(d,n-2)\\r\ {\rm even}}}E_r
 \notag\\
 &+\sum_{\substack{3\le r\le\min(d,n-3)\\r\ {\rm odd}}}C_r
 \notag\\
 &+\one{d\ge n-1}\bigl(-E_{n-1}+2C_{n-1}\bigr)
 \notag\\
 &-2\one{d=n}C_n.
 \label{eq:even-potential-en}
\end{align}

\section[Odd-dimensional case]{Odd-dimensional case}

{The equality branch established in the odd-dimensional work \cite{XiangEtAl2026Odd} and the inequality
proved below combine into the following local dichotomy.}

\begin{theorem}[Odd local dichotomy]
\label{thm:odd-main-en}
Let \(n=2K+1\ge3\).  Assume \cref{eq:CY-en,eq:generic-en}, and assume that star-cluster
arms use distinct positive coordinate directions.  For every \(1\le d\le n\) and every
local ideal \(\calI\subseteq2^D\),
{\begin{equation}
 \omega^{\odd}_{0,\Delta^{(n)}}(\vec q_d)
 =\begin{cases}
 1,&\Delta^{(n)}\in G(\vec q_d),\\
 \le0,&\Delta^{(n)}\notin G(\vec q_d).
 \end{cases}
 \label{eq:odd-theorem-en}
\end{equation}
}
\end{theorem}

\begin{proof}
Suppose first that \(d<n\).  If \(d\le n-2\), the ranges in
\cref{eq:odd-potential-en} agree with those in \cref{eq:Phi-en}.  If \(d=n-1\), then \(d\)
is even, so the largest odd reverse rank in the standard statistic is still \(n-2\).  Thus
\[
 \omega^{\odd}_{0,\Delta^{(n)}}(\vec q_d)=\Phi
\]
in both cases, and the claim follows from
\cref{eq:admissible-lower-en,eq:standard-dichotomy-en}.

Now let \(d=n\).  By \cref{eq:admissible-full-en}, the four admissible ideals are
\[
 2^D,\qquad2^D\setminus\{D\},\qquad\varnothing,\qquad\{\varnothing\}.
\]
The last two have no contribution in reverse ranks \(r<n\), so only the vacuum term \(1\)
remains.  For the first two, all relevant layers \(1\le r\le n-1\) are complete and the
non-vacuum contribution is
\[
 \sum_{r=1}^{n-1}(-1)^{r+1}\binom nr=0
\]
because \(n\) is odd.  Adding the vacuum term again gives \(1\).

It remains to consider a full-dimensional non-admissible ideal.  It is neither void nor the
one-origin ideal, so it contains a nonempty set and hence an occupied singleton.  It is also
not one of the two top extremal ideals, so a coatom is missing.  In fact, one can choose the
same \(j\in D\) so that
\begin{equation}
 \{j\}\in\calI,
 \qquad
 D\setminus\{j\}\notin\calI.
 \label{eq:odd-special-j-en}
\end{equation}
Otherwise every occupied singleton would have its corresponding coatom occupied.  Choose
one occupied \(\{j_0\}\).  The coatom \(D\setminus\{j_0\}\), by downward closure, forces every
other singleton to be occupied; the assumption then forces every coatom to be occupied, a
contradiction.

{The actual positive and negative classes are
\begin{align*}
 \calP_{\odd}
 &:={}
 \{A\in\calI:\rho(A)=1\}
 \notag\\[-0.2em]
 &\quad\cup
 \left\{A\in\calI:
 \substack{A\text{ saturated},\\3\le\rho(A)\le n-2\text{ odd}}
 \right\},\\
 \calN_{\odd}
 &:={}
 \{A\in\calI:2\le\rho(A)\le n-1\text{ even}\}.
\end{align*}
}
Thus
\[
 \omega^{\odd}_{0,\Delta^{(n)}}(\vec q_n)
 =1+|\calP_{\odd}|-|\calN_{\odd}|.
\]
The case check in \cref{thm:flip-en} remains valid for these truncated classes, so
\(\tau_j\) injects \(\calP_{\odd}\) into \(\calN_{\odd}\).  Moreover, since \(n\) is odd,
\(\rho(\{j\})=n-1\) is even, and hence \(\{j\}\in\calN_{\odd}\).  Its unique
\(\tau_j\)-preimage in the full Boolean lattice is \(\varnothing\), whose reverse rank \(n\)
does not belong to \(\calP_{\odd}\).  Consequently,
\[
 |\calN_{\odd}|\ge|\calP_{\odd}|+1.
\]
The extra negative object cancels the vacuum term, giving
\(\omega^{\odd}_{0,\Delta^{(n)}}(\vec q_n)\le0\).
\end{proof}

\section[Even-dimensional case]{Even-dimensional case}

{The equality branch established in the even-dimensional work \cite{FengChenZhang2025All} and the inequality
proved below combine into the following local dichotomy.}

{Introduce the truncated statistic
\begin{equation}
 \begin{aligned}
 H_{n,d}
 :={}&E_1
 -\sum_{\substack{2\le r\le\min(d,n-2)\\r\ {\rm even}}}E_r\\
 &+\sum_{\substack{3\le r\le\min(d,n-3)\\r\ {\rm odd}}}C_r\\
 &+\one{d\ge n-1}C_{n-1}.
 \end{aligned}
 \label{eq:H-en}
\end{equation}
}
By \cref{eq:even-potential-en},
\begingroup\small
\begin{equation}
 \begin{aligned}
  \omega^{\even}_{0,\Delta^{(n)}}(\vec q_d)
  &=H_{n,d}, && d\le n-2,\\
  \omega^{\even}_{0,\Delta^{(n)}}(\vec q_{n-1})
  &=H_{n,n-1}\\
  &\quad-(E_{n-1}-C_{n-1}),\\
  \omega^{\even}_{0,\Delta^{(n)}}(\vec q_n)
  &=1+H_{n,n}\\
  &\quad-(E_{n-1}-C_{n-1})-2C_n.
 \end{aligned}
 \label{eq:even-decomposition-en}
\end{equation}
\endgroup

\begin{theorem}[Even local dichotomy]
\label{thm:even-main-en}
Let \(n=2K\ge4\), with the same genericity and distinct-arm assumptions as in
\cref{thm:odd-main-en}.  For every \(1\le d\le n\) and every local ideal
\(\calI\subseteq2^D\),
{\begin{equation}
 \omega^{\even}_{0,\Delta^{(n)}}(\vec q_d)
 =\begin{cases}
 1,&\Delta^{(n)}\in G(\vec q_d),\\
 \le0,&\Delta^{(n)}\notin G(\vec q_d).
 \end{cases}
 \label{eq:even-theorem-en}
\end{equation}
}
\end{theorem}

\begin{proof}
We use four steps.

\emph{Step 1: \(d\le n-2\).}
Here \(H_{n,d}=\Phi\).  The claim follows directly from
\cref{eq:admissible-lower-en,eq:standard-dichotomy-en}.

\emph{Step 2: \(d=n-1\).}
Now \(d=n-1\) is odd and \(H_{n,n-1}=\Phi\).  Since
\(C_{n-1}\le E_{n-1}\),
\begin{equation}
 \omega^{\even}_{0,\Delta^{(n)}}(\vec q_{n-1})
 =\Phi+C_{n-1}-E_{n-1}\le\Phi.
 \label{eq:even-nminus1-en}
\end{equation}
Thus every non-admissible ideal has nonpositive potential.  In either admissible extremal
ideal, the only center of reverse rank \(n-1=d\) is \(A=\varnothing\), and
\(C_{n-1}=E_{n-1}=1\).  The correction cancels and the value remains \(1\).  Notice that the
center is the empty set here; only when \(d=n\) does reverse rank \(n-1\) consist of
singletons.

\emph{Step 3: non-admissible ideals with \(d=n\).}
Choose \(j\) with \(D\setminus\{j\}\notin\calI\) by
\cref{lem:missing-coatom-en}.  Define
{\begin{align}
 \calP_H
 &:={}
 \{A\in\calI:\rho(A)=1\}
 \notag\\[-0.2em]
 &\quad\cup
 \left\{A\in\calI:
 \substack{A\text{ saturated},\\3\le\rho(A)\le n-1\text{ odd}}
 \right\},
 \label{eq:PH-en}\\
 \calN_H
 &:={}
 \{A\in\calI:2\le\rho(A)\le n-2\text{ even}\}.
 \label{eq:NH-en}
\end{align}
}
Then \(H_{n,n}=|\calP_H|-|\calN_H|\).  We check the flip by rank.
\begin{enumerate}
  \item A reverse-rank-one object is sent to reverse rank two exactly as in
        \cref{thm:flip-en}.
  \item For a saturated odd reverse rank \(3\le r\le n-3\), a downward flip is legal by
        downward closure and an upward flip by saturation.  The image rank belongs to
        \(2,4,\ldots,n-2\).  This class is empty when \(n=4\).
  \item If \(\rho(A)=n-1\), then \(A=\{k\}\) is a saturated singleton.  If \(k\ne j\),
        saturation gives
        \(\tau_j(A)=\{j,k\}\in\calI\), of reverse rank \(n-2\).  The only possible exception is
        \[
          A=\{j\},\qquad \tau_j(A)=\varnothing,\qquad
          \rho(\varnothing)=n,
        \]
        and \(\varnothing\notin\calN_H\).
\end{enumerate}
We therefore have the precise injection
\begin{equation}
 \tau_j:\calP_H\setminus\{\{j\}\}\hookrightarrow\calN_H,
 \label{eq:even-endpoint-injection-en}
\end{equation}
where no object is removed if \(\{j\}\notin\calP_H\).  Hence
\begin{equation}
 H_{n,n}\le1.
 \label{eq:H-upper-en}
\end{equation}

If \(C_n=1\), the highest \((n+1)\)-box cluster contributes \(-2\), while
\(E_{n-1}-C_{n-1}\ge0\).  By
\cref{eq:even-decomposition-en,eq:H-upper-en},
\[
 \omega^{\even}_{0,\Delta^{(n)}}(\vec q_n)\le1+H_{n,n}-2\le0.
\]

If \(C_n=0\), then \(C_{n-1}=0\).  Indeed, if a singleton \(\{k\}\) were saturated, then every
\(\{k,\ell\}\), \(\ell\ne k\), would be occupied.  Downward closure would force all
singletons, as well as the origin, to be occupied; the origin would then be saturated and
\(C_n=1\), a contradiction.  Thus the only possible endpoint in
\cref{eq:even-endpoint-injection-en} is absent, and the full injection
\(\calP_H\hookrightarrow\calN_H\) gives \(H_{n,n}\le0\).  A non-admissible ideal is neither
void nor the one-origin ideal, so it contains a nonempty set and hence at least one singleton:
\(E_{n-1}\ge1\).  Consequently,
\[
 \omega^{\even}_{0,\Delta^{(n)}}(\vec q_n)=1+H_{n,n}-E_{n-1}\le0.
\]

\emph{Step 4: the four admissible ideals with \(d=n\).}
The void ideal has all \(E_r,C_r\) equal to zero.  The one-origin ideal has \(E_n=1\), but
\cref{eq:even-potential-en} has no \(E_n\) term and \(C_n=0\).  Both therefore leave only the
vacuum contribution \(1\).  For \(2^D\) and \(2^D\setminus\{D\}\),
\(E_r=\binom nr\) and every relevant \(C_r=\binom nr\).  The two highest layers satisfy
\[
 \begin{aligned}
 -E_{n-1}+2C_{n-1}&=+\binom n{n-1},\\
 1-2C_n&=-1=(-1)^{n+1}\binom nn.
 \end{aligned}
\]
Thus all terms reconstruct
\[
 \sum_{r=1}^{n}(-1)^{r+1}\binom nr=1.
\]
\end{proof}

\begin{remark}[Why the parity cases do not collapse to one line]
The middle ranks are governed by the same matching in both parities.  The difference is
confined to the full-dimensional bottom endpoint.  If \(n\) is odd, a singleton has even
reverse rank \(n-1\) and supplies an extra unmatched negative object.  If \(n\) is even, a
saturated singleton has odd reverse rank and may be an unmatched positive object because it
flips to \(\varnothing\) of reverse rank \(n\), outside \(\calN_H\).  The special even
correction \(-E_{n-1}+2C_{n-1}-2C_n\) absorbs precisely this one endpoint.  Thus the correct
organization is one common matching theorem followed by two parity-dependent endpoint
corollaries.
\end{remark}

\section{Relation to the prior works}

\subsection{From computation to proof}

{The odd-dimensional paper~\cite{XiangEtAl2026Odd} analytically proved the equality
branch.  For the inequality branch, it enumerated all}
\[
 M(1)+\cdots+M(5)=3+6+20+168+7581=7778
\]
{ideals for \(n=5\), and used sequential-growth Monte Carlo tests for \(n=7,9\).  The
even-dimensional paper~\cite{FengChenZhang2025All} analytically proved its equality branch.
For the inequality branch, it enumerated}
\[
 M(1)+\cdots+M(6)=7836132
\]
configurations for \(n=6\) and sampled in dimension eight.  Sequential-growth sampling starts
at the void ideal and, at each step, chooses uniformly among the currently addable boxes.
At fixed cardinality this is generally not the uniform distribution on ideals: the probability
of an ideal depends on the number of growth paths leading to it and on the sizes of the
addable sets along those paths.  Sampling is useful for counterexample searches but does not
replace a proof.

{The explicit matching proves both inequality branches in every corresponding dimension.
The earlier computations remain useful independent checks, but they are no longer needed for
the local inequalities.}

\subsection{Local-to-global connection}

{\Cref{thm:odd-main-en,thm:even-main-en} prove the local inequality branches in the
odd- and even-dimensional cases.  Together with the established equality branches, they give
the complete local dichotomies.}  {Here the \emph{bisection step} means the
decomposition of an arbitrary partition into a background part and one local hypercube.}
The remaining reduction uses cancellation between contributions from different clusters,
the vanishing of the background contribution when \(d<n\), and full-dimensional induction
and translation.
Together with Lemmas~1--4 of
\cite{XiangEtAl2026Odd}, the present theorem replaces the finite-enumeration and Monte Carlo
{input for the odd-dimensional inequality.}

Section~V of the even-dimensional paper \cite{FengChenZhang2025All} uses the same local-to-global strategy: it separates
the background from the local hypercube, then proves that the background contribution
vanishes and cancels the remaining contributions from different clusters.
{\Cref{thm:even-main-en} supplies the analytic local
input for the even-dimensional inequality,}
while the highest \(n\)-box and \((n+1)\)-box clusters retain the treatment of the cited work.

\section{Combinatorial extensions}

{The coordinate flip may also be useful for weighted layer statistics or for related
matchings on other graded posets.  We leave these extensions, as well as the treatment of
non-generic weights with projection multiplicities, for future work.}

\section{Conclusion}

{We have proved the local inequalities in both the odd- and even-dimensional cases.  The
proof has two ingredients.  First, genericity turns the local pole order into a
multiplicity-one signed count on a Boolean lattice.  Second, after choosing a missing coatom,
the flip
\[
 A\longmapsto A\triangle\{j\}
\]
pairs every intermediate-rank positive contribution with a distinct negative contribution.
Downward moves follow from the ideal property, upward moves from saturation, and injectivity
from the fact that the flip is an involution.

Only the full-dimensional endpoint depends on parity.  In odd dimension an unmatched
negative singleton cancels the vacuum contribution.  In even dimension the possible
unmatched positive singleton is offset by the highest-rank terms already present in the
charge function.  This completes the local input required by the reduction arguments in the
prior works.
}

\section*{Acknowledgments}
The authors thank Hao Feng and Shang Xiang
for helpful discussions.  K.Z. (Hong Zhang) is supported by a classified fund from Shanghai city.

\bibliographystyle{unsrtnat}
\bibliography{ZCZ}

\begin{thebibliography}{10}
\providecommand{\natexlab}[1]{#1}
\providecommand{\url}[1]{\texttt{#1}}
\expandafter\ifx\csname urlstyle\endcsname\relax
  \providecommand{\doi}[1]{doi: #1}\else
  \providecommand{\doi}{doi: \begingroup \urlstyle{rm}\Url}\fi

\bibitem[Xiang et~al.(2026)Xiang, Feng, Zhuo, Chen, and
  Zhang]{XiangEtAl2026Odd}
Shang Xiang, Hao Feng, Keyou Zhuo, Tian-Shun Chen, and Kilar Zhang.
\newblock Charge functions for odd dimensional partitions.
\newblock \emph{Journal of High Energy Physics}, 2026\penalty0 (5):\penalty0
  141, 2026.
\newblock \doi{10.1007/JHEP05(2026)141}.
\newblock URL \url{https://arxiv.org/abs/2512.07758}.

\bibitem[Feng et~al.(2025)Feng, Chen, and Zhang]{FengChenZhang2025All}
Hao Feng, Tian-Shun Chen, and Kilar Zhang.
\newblock Charge functions for all dimensional partitions, 2025.
\newblock URL \url{https://arxiv.org/abs/2512.24343}.

\bibitem[Dedekind(1897)]{Dedekind1897}
Richard Dedekind.
\newblock {\"U}ber zerlegungen von zahlen durch ihre gr{\"o}ssten gemeinsamen
  theiler.
\newblock In \emph{Fest-Schrift der Herzoglichen Technischen Hochschule
  Carolo-Wilhelmina}, pages 1--40. Vieweg+Teubner Verlag, 1897.
\newblock ISBN 9783663072249.
\newblock \doi{10.1007/978-3-663-07224-9_1}.
\newblock Reprinted in Gesammelte mathematische Werke, Vol. 2, pp. 103--148.

\bibitem[Chen et~al.(2026)Chen, Feng, Wang, Chen, and Zhang]{Chen:2026cfo}
Tian-Shun Chen, Hao Feng, Haozhe Wang, Chian-Shu Chen, and Kilar Zhang.
\newblock {Finite-n Estimate of Dedekind Numbers by Layer-Ratio Monte Carlo}.
\newblock 6 2026.

\bibitem[Garsia and Milne(1981)]{GarsiaMilne1981}
Adriano~M. Garsia and Stephen~C. Milne.
\newblock Method for constructing bijections for classical partition
  identities.
\newblock \emph{Proceedings of the National Academy of Sciences of the United
  States of America}, 78\penalty0 (4):\penalty0 2026--2028, 1981.
\newblock \doi{10.1073/pnas.78.4.2026}.
\newblock URL \url{https://doi.org/10.1073/pnas.78.4.2026}.

\bibitem[Forman(1998)]{Forman1998}
Robin Forman.
\newblock Morse theory for cell complexes.
\newblock \emph{Advances in Mathematics}, 134\penalty0 (1):\penalty0 90--145,
  1998.
\newblock \doi{10.1006/aima.1997.1650}.
\newblock URL \url{https://doi.org/10.1006/aima.1997.1650}.

\bibitem[Kozlov(2008)]{Kozlov2008}
Dmitry Kozlov.
\newblock \emph{Combinatorial Algebraic Topology}, volume~21 of
  \emph{Algorithms and Computation in Mathematics}.
\newblock Springer, Berlin, Heidelberg, 2008.
\newblock \doi{10.1007/978-3-540-71962-5}.
\newblock URL \url{https://doi.org/10.1007/978-3-540-71962-5}.

\bibitem[Kruskal(1963)]{Kruskal1963}
Joseph~B. Kruskal.
\newblock The number of simplices in a complex.
\newblock In Richard Bellman, editor, \emph{Mathematical Optimization
  Techniques}, pages 251--278. University of California Press, Berkeley, 1963.
\newblock \doi{10.1525/9780520319875-014}.
\newblock URL \url{https://doi.org/10.1525/9780520319875-014}.

\bibitem[Katona(1968)]{Katona1968}
Gyula O.~H. Katona.
\newblock A theorem of finite sets.
\newblock In \emph{Theory of Graphs: Proceedings of the Colloquium Held at
  Tihany, Hungary, September 1966}, pages 187--207. Akad\'emiai Kiad\'o,
  Budapest, 1968.
\newblock URL \url{https://www.renyi.hu/~ohkatona/publang.html}.

\bibitem[Stanley(2012)]{Stanley2012}
Richard~P. Stanley.
\newblock \emph{Enumerative Combinatorics}, volume~1 of \emph{Cambridge Studies
  in Advanced Mathematics, vol. 49}.
\newblock Cambridge University Press, Cambridge, 2 edition, 2012.
\newblock \doi{10.1017/CBO9781139058520}.
\newblock URL \url{https://doi.org/10.1017/CBO9781139058520}.

\end{thebibliography}

\end{document}